\documentclass[12pt]{amsart}
\usepackage{amssymb,latexsym}
\usepackage{amsfonts}
\usepackage{amsmath}
\usepackage[colorlinks,linkcolor=blue,anchorcolor=blue,citecolor=blue]{hyperref}
\usepackage{algorithm}
\usepackage{enumerate}
\usepackage{algpseudocode}
\usepackage{verbatim}
\usepackage{graphicx}
\newtheorem{theorem}{Theorem}
\newtheorem{lemma}[theorem]{Lemma}
\newtheorem{example}[theorem]{Example}

\newtheorem{corollary}[theorem]{Corollary}
\newtheorem{proposition}[theorem]{Proposition}
\newtheorem{definition}[theorem]{Definition}

\newtheorem{remark}[theorem]{Remark}

\def\F{\mathbb{F}}

\def\Tr{\text{\rm Tr}}
\def\wt{\text{\rm wt}}

\def\d{\textrm{d}}
\def\x{\textrm{x}}
\def\D{\textrm{D}}
\def\c{\textrm{c}}
\def\Prj{\textrm{Prj}}
\def\a{\textrm{a}}

\numberwithin{equation}{section}
\numberwithin{theorem}{section}

\begin{document}

\title[Weight distributions and weight hierarchies ]{Weight distributions and weight hierarchies of two classes of binary linear codes}
\markright{Weight hierarchies and weight distributions}

\author{Fei Li}
\address{Faculty of School of Statistics and Applied Mathematics,
Anhui University of Finance and Economics, Bengbu,  Anhui Province, 233030, China}
\email{cczxlf@163.com; 120110029@aufe.edu.cn}

\author{Xiumei Li}
\address{School of Mathematical Sciences, Qufu Normal University, Qufu Shandong, 273165, China}
\email{lxiumei2013@qfnu.edu.cn}



\begin{abstract}
Linear codes
with a few weights can be applied to communication, consumer electronics and data storage system. In addition, the weight hierarchy of linear codes has many applications such as
 on the type II wire-tap channel, dealing
with $t$-resilient functions and trellis or branch complexity of
linear codes and so on. In this paper,
 we first present a formula for computing the weight hierarchies
of linear codes constructed by the generalized method of defining sets.
Then, we construct two classes of binary linear codes with a few weights and
determine their weight distributions and weight hierarchies completely.
Some codes of them can be used in secret sharing schemes.

\end{abstract}

\maketitle

\section{Introduction}
\label{intro}
For a prime number $p$ and a positive integer $s$, let $ \mathbb{F}_{p^s} $ be the finite field with $ p^{s} $
elements and $\mathbb{F}_{p^s}^{*}$ be its multiplicative group.

An $[n,k,d]$ $p$-ary linear code $C$ is a $k$-dimensional subspace of $ \mathbb{F}_{p}^{n} $ with minimum (Hamming) distance $d$.
For $i\in\{1,2,\cdots,n\}$, denote by $A_{i}$ the number of codewords  in $C$ with Hamming weight $i$. The sequence $(1,A_{1},\cdots,A_{n})$
is called the weight distribution of $C$ and the polynomial $1+A_{1}x+A_{2}x^{2}+\cdots+A_{n}x^{n}$ is called the weight enumerator of $C$.
If the sequence $(A_{1},\cdots,A_{n})$ has $t$ nonzero $A_{i}$ with $ i=1,2,\cdots n$, then the code $C$ is called a $t$-weight code.
In coding theory, the weight distribution of linear codes is a classical research topic and attracts much attention. Furthermore,
linear codes with a few weights have important applications in authentication codes \cite{8DH07},
association schemes \cite{4CG84}, secret sharing \cite{21YD06} and strongly regular graphs \cite{5CK86}.

Motivated by applications to cryptography, generalized Hamming weight of linear codes was introduced in 1970s \cite{HKM77,K78}. For linear code $C$, denote by
$ [C,r]_{p} $ the set of all its $\mathbb{F}_{p}$-vector subspaces with dimension $r$.
For $ V \in [C,r]_{p}$, define $ \textrm{Supp}(V)=\cup_{c\in V}\textrm{Supp}(c)$, where $\textrm{Supp}(c)$ is the set of coordinates where $c$ is nonzero, that is,
$$ \textrm{Supp}(V)=\{i:1\leq i\leq n, c_i\neq 0 \ \ \textrm{for some $c=(c_{1}, c_{2}, \cdots , c_{n})\in V$}\}.$$

\begin{definition}
Let $C$ be an $[n, k, d]$ linear code over $\F_p$. For $1 \leq r \leq k$,
$$
d_{r}(C)=\min\{|\textrm{Supp}(V)|:V\in [C,r]_{p}\}
$$
is called the $r$-th generalized Hamming weight (GHW) of $C$ and $\{d_r(C) : 1 \leq r \leq k\}$ is
called the weight hierarchy of $C$.
\end{definition}

GHW can be thought of an extension of the minimum distance $d = d_1(C)$ and has become an important research
object in coding theory after Wei's paper \cite{20WJ91} in 1991.
A detailed overview on the results of GHW up to 1995 was given in \cite{19TV95}.
The weight hierarchy of linear codes has many applications such as
completely characterizing the performance of the code on the type II wire-tap channel, dealing
with $t$-resilient functions and trellis or branch complexity of
linear codes and so on \cite{9HK92,19TV95,20WJ91}.
Much is known about weight hierarchy for several classes of codes: algebraic geometric codes,
BCH codes, Reed-Muller codes, cyclic codes \cite{1BL14,3CC97,11HP98,14JL97,21XL16,22YL15}.
Recently, there are research results about weight hierarchy of some classes of linear codes \cite{2B19,13JF17,18LF17,19LW19}.

The rest of this paper is organized as follows.
In Sec. 2, we introduce a generalized method of constructing linear code by defining sets and
give a corresponding formula for computing the generalized Hamming weights.
In Sec. 3, we construct binary linear codes with a few weights and determine their weight distributions and weight hierarchies completely.
In Sec. 4, we summarize the paper.

\section{Preliminaries}
\subsection{A generic construction of linear codes}

Ding et al. \cite{6DD14} proposed a generic construction of linear codes as below.
Denote by $\mathrm{Tr}$ the trace function from $\mathbb{F}_{p^s}$ to $\mathbb{F}_p$.
Let $ D= \{d_{1},d_{2},\cdots,d_{n}\}$ be a subset of $\mathbb{F}_{p^s}^{\ast}.$
A $p$-ary linear code of length $n$ is defined as follows:
\begin{eqnarray}\label{defcode1}
         C_{D}=\{\left( \Tr(xd_1), \Tr(xd_2),\ldots, \mathrm{Tr}(xd_{n})\right):x\in \mathbb{F}_{p^{s}}\},
\end{eqnarray}
and $D$ is called the defining set of $C_{D}$. By choosing some proper defining sets, some optimal linear
codes with a few weights can be constructed \cite{5DJ15,25DD15,9DL16,16KY19,LL18,26TX17,19TXF17,23YY17,24ZL16}.

Li et al. generalized the above constructing method and also constructed some linear codes with a few weights \cite{LBY19,LY17}.
The generalized method is as follows.
For a positive integer $e,$ let $ \textrm{D}= \{\d_{1},\d_{2},\cdots,\d_{n}\}$
be contained in $\mathbb{F}_{p^{s}}^{e}\backslash\{(0,0,\cdots,0)\}. $
For $\textrm{u}=(u_{1},u_{2},\cdots,u_{e})$, $\textrm{v}=(v_{1},v_{2},\cdots,v_{e})$,
denote by $\textrm{u}\cdot \textrm{v}$ the ordinary inner product of $\textrm{u}$ and $\textrm{v}$, that is,
$$
\textrm{u}\cdot \textrm{v}=u_{1}v_{1}+u_{2}v_{2}+\cdots+u_{e}v_{e}.
$$

Define
a $p$-ary linear code $ C_{\D} $ with length $ n $ as follows:
\begin{eqnarray}\label{defcode1}
         C_{\D}=\{\left( \mathrm{Tr}(\x\cdot \d_1), \mathrm{Tr}(\x\cdot \d_2),\ldots, \mathrm{Tr}(\x\cdot \d_{n})\right):\x\in \mathbb{F}_{p^{s}}^{e}\}.
\end{eqnarray}
Here, $ \D $ is also called the defining set. Using this method, some classes of linear codes with few weights
were obtained by choosing proper defining sets \cite{1AK17,JL19}.

\subsection{A formula for computing $ d_{r}(C_{\D})$}

For the linear code $ C_{\D}$ in \eqref{defcode1}, a general formula in Proposition \ref{pro:d_r} will be employed to calculate the generalized Hamming weight $ d_{r}(C_{\D})$.
Indeed, the result of Proposition \ref{pro:d_r} can be regarded as a generalization of \cite[Theorem 1]{18LF17}.

\begin{proposition}\label{pro:d_r}
For each $ r $ and $ 1\leq r \leq es$, if the dimension of $ C_{\D} $
is $ es$, then
\begin{equation*}\label{eq:d_r}
d_{r}(C_{\D})= n-\max\big\{|\D \cap H|: H \in [\mathbb{F}_{p^{s}}^{e},es-r]_{p}\big\}.
\end{equation*}
\end{proposition}

\begin{proof}
Let $ f $ be a map from $ \mathbb{F}_{p^{s}}^{e} $ to $ \mathbb{F}_{p}^{n}$ defined by
$$
f(\x)=\big( \mathrm{Tr}(\x\cdot \d_1), \mathrm{Tr}(\x\cdot \d_2),\ldots, \mathrm{Tr}(\x\cdot \d_{n})\big).
$$
Obviously, $f$ is an $\F_p$-linear homomorphism and the image $f(\mathbb{F}_{p^{s}}^{e})$ is $C_{\D}$. By hypothesis, we know that the dimension of $ C_{\D} $
is $es$. Hence, $f$ is injective.

For any $r$-dimensional subspace $ U_{r} \in [C_{\D}, r]_{p}$, denote by $H_{r} = f^{-1}(U_{r})$
the inverse image of $U_{r}$. So, $H_{r}$ is also an $r$-dimensional
subspace. Let $ \{ \beta_{1},\beta_{2},\ldots ,\beta_{r}\}$ be an $\mathbb{F}_{p}$-basis of $ H_{r}$. By definition,
$$
d_{r}(C_{\D})= n-\max\{N(U_{r}): U_{r} \in [C_{\D}, r]_{p}\},
$$
where
\begin{align*}
N(U_{r})&=\Big|\{i: 1\leq i \leq n, c_{i}=0\ \ \textrm{for any\ } \c=(c_{1},c_{2}, \ldots, c_{n}) \in U_{r}\}\Big| \\
&=\Big|\{i: 1\leq i \leq n, \Tr(\beta \cdot \d_{i})=0\ \ \textrm{for any\ } \beta \in H_{r}\}\Big|\\
&=\Big|\{i: 1\leq i \leq n, \Tr(\beta_j \cdot \d_{i})=0, j =1,2,\cdots,r\}\Big|.
\end{align*}
Then, by the orthogonal property of additive characters, we have
\begin{align*}
N(U_{r})
&=\frac{1}{p^{r}}\sum_{i=1}^{n}\sum_{x_{1}\in \mathbb{F}_{p}}\zeta_{p}^{\mathrm{Tr}\big(x_{1}(\beta_{1}\cdot \d_{i})\big)}\ldots \sum_{x_{r}\in \mathbb{F}_{p}}\zeta_{p}^{\mathrm{Tr}\big(x_{r}(\beta_{r}\cdot \d_{i})\big)}  \\
&=\frac{1}{p^{r}}\sum_{i=1}^{n}\sum_{x_{1},\ldots,x_{r}\in \mathbb{F}_{p}}\zeta_{p}^{\mathrm{Tr}\big( (\beta_{1}x_{1}+\ldots+\beta_{r}x_{r})\cdot \d_{i}\big)}    \\
&=\frac{1}{p^{r}}\sum_{i=1}^{n}\sum_{\beta\in H_{r}}\zeta_{p}^{\mathrm{Tr}(\beta\cdot \d_{i})}.
\end{align*}

For any $k$-dimensional subspace $H$ of $\mathbb{F}_{p^{s}}^{e}$,
let
$$
H^{\perp}=\{\textrm{v}\in \mathbb{F}_{p^{s}}^{e}: \mathrm{Tr}(\textrm{u}\cdot \textrm{v})=0 \ \textrm{for any}\ \textrm{u} \in H\}.
$$
We call $H^{\perp}$ the dual space of $H$. Taking a $\mathbb{F}_{p}$-basis $\alpha_{1},\alpha_{2},\cdots,\alpha_{es}$ of $\mathbb{F}_{p^{s}}^{e}$ and $\gamma_1,\gamma_2,\cdots,\gamma_k$ of $H$, then the matrix $\big(\mathrm{Tr}(\alpha_{i}\cdot \alpha_{j})\big)$ is an invertible matrix. For any $\textrm{v}=\sum\limits_{i=1}^{es}x_i\alpha_i\in H^\perp$ with $x_i\in\F_p$, we have $\Tr(\gamma_j\cdot\textrm{v})=\sum\limits_{i=1}^{es}x_i\Tr(\gamma_j\cdot\alpha_i)=0, j=1,2,\cdots,k$.
This is a system of linear equations of $x_1,x_2,\cdots,x_{es}$ and the rank of its coefficient matrix $(\Tr(\gamma_j\cdot\alpha_i))$ is equal to $k$.
Thus, $ \dim_{\mathbb{F}_{p}}(H)+\dim_{\mathbb{F}_{p}}(H^{\bot})=es$ and $(H^{\bot})^{\bot}=H$. So,
For $\textrm{y}\in \mathbb{F}_{p^{s}}^{e}$,
$$
\sum_{\beta\in H_{r}}\zeta_{p}^{\mathrm{Tr}(\beta\cdot \textrm{y})}=\left\{\begin{array}{ll}
|H_{r}|, & \textrm{if\ } \ \textrm{y} \in H_{r}^{\bot}, \\
0, & \textrm{otherwise\ }.
\end{array}
\right.
$$
Hence, we have
$$
N(U_{r})
=\frac{1}{p^{r}}\sum_{y\in D \cap H_{r}^{\bot}} |H_{r}|=|D \cap H_{r}^{\bot}|.
$$
So, the desired result follows from the fact that there is a bijection
between $[\mathbb{F}_{p^{s}}^{e},r]_{p}$ and $[\mathbb{F}_{p^{s}}^{e},es-r]_{p}$.
We complete the proof.
\end{proof}

\section{Two classes of binary linear codes}

In this section, we construct two classes of binary linear codes and determine their parameters.
We firstly present a few more auxiliary results which will be needed in proving our main results.

Let $l$ be a prime such that $2$ is a primitive root modulo $l^{m}$.
Here and after, let $p=2, q=2^{\phi(l^{m})},$  where $\phi$ is the Euler phi-function.
Let $\gamma$ be a fixed primitive element of $\mathbb{F}_{q}^*$ and $\chi_{1}$ be
the \textit{canonical additive character} over $\mathbb{F}_{q}$, then for any $x\in\F_q$,
$\chi_{1}(x)=(-1)^{\Tr(x)}$. Set $\alpha=\gamma^{\frac{q-1}{l^{m}}}$.
See \cite{16LN97} for more information about additive characters over finite fields.

\subsection{An exponential sum}\label{subsec:3.1}

For any $ a, b\in \mathbb{F}_{q}$, we define
$$
S(a,b)=\sum_{x\in \mathbb{F}_{q}^{*}}\chi_1\left(ax^{\frac{q-1}{l^{m}}}+bx\right)\ \ \textrm{and} \ \ S(a)=\sum_{i=0}^{l^{m}-1}\chi_{1}(a\alpha^{i}).
$$

Since $2$ is a primitive root modulo $l^{m},$ we have that $\mathbb{F}_{q}=\mathbb{F}_{2}(\alpha).$
Expressing $u\in\mathbb{F}_{q}$  in the basis
$\{\alpha, \alpha^{2},\ldots,\alpha^{\phi(l^{m})}\}$, say
$$u=\sum_{i=1}^{\phi(l^{m})}a_{j}\alpha^{j},$$ where $a_{j}\in\mathbb{F}_{2}.$
For $i=0,1,\ldots,l^{m-1}-1$, we denote by
$\textrm{u}^{(i)}$ the following sub-vector of length $l-1$ of the coordinate vector $\textrm{u}=(a_{1},\ldots,a_{\phi(l^{m})})$ of
$u$:
$$\textrm{u}^{(i)}=(a_{l^{m-1}-i},a_{2l^{m-1}-i},\ldots,a_{(l-1)l^{m-1}-i}).$$

Recall that $\wt(\x)$ is the Hamming weight of binary vector $\x$. For $a\in \mathbb{F}_{q}^{\ast}$, we define two subsets of $\mathbb{F}_{q}^{\ast}$ as follows.
$$
E_{a}=\{u\in \mathbb{F}_{q}^{\ast}: \ 2|\wt\big((a u^{-\frac{q-1}{l^{m}}})^{(0)}\big) \},
$$
$$
O_{a}=\{u\in \mathbb{F}_{q}^{\ast}: \ 2\nmid \wt\big((a u^{-\frac{q-1}{l^{m}}})^{(0)}\big) \}.
$$
In subsection 3.3, we will give the cardinal numbers of $E_{a}$ and $O_{a}$.

Now, we give some results on $S(a,b)$ and $S(a)$ in the following lemmas, which will
play important roles in settling the parameters of our codes.

\begin{lemma}\cite[Lemma 4]{LL18}\label{lem:S(a,b)}
 Let $ a\in\F_q^*, b \in \mathbb{F}_{q}$ and let $ c=ab^{-\frac{q-1}{l^{m}}} $ if $b\neq0.$ Then,

$$
S(a,b)=\left\{\begin{array}{ll}
\frac{q-1}{l^{m}}S(a), & \textrm{if\ } \ b=0, \\
(-1)^{\wt(\c^{(0)})}\sqrt{q}-\frac{\sqrt{q}+1}{l^{m}}S(a), & \textrm{if\ } \ b\neq0.
\end{array}
\right.
$$
\end{lemma}

\begin{lemma}\cite[Theorem 1]{26LL18}\label{lem:S(a)}
 Let $a\in\mathbb{F}_{q}$. Then
$$ S(a)=\sum_{i=0}^{l^{m-1}-1}(-1)^{\wt(\a^{(i)})}(l-2\wt(\a^{(i)})).$$
\end{lemma}

\begin{lemma}\cite[Theorem 3]{26LL18}\label{lem:3}
 The value set of $S(a)$, as $a$ runs over $\mathbb{F}_{q}^{*}$, is
$$ \{l^{m}-4j|j=1,\ldots,\frac{l^{m-1}(l-1)}{2}\}.$$
\end{lemma}

\subsection{Our defining sets and some auxiliary results}\label{subsec:3.2}

In this paper, we choose a defining set contained in $\mathbb{F}_{q}^{2} $ as follows.
For $a\in\mathbb{F}_{q}^{\ast}, b\in\mathbb{F}_{q}$, set
\begin{align*}
\D_{(a,b)}&=\{(x,y)\in \mathbb{F}_{q}^{2}\setminus\{(0,0)\}: \mathrm{Tr}(ax^{\frac{q-1}{l^{m}}}+by)=0\}
=\{\d_1,\d_2,\ldots,\d_n\}
\end{align*}
and the corresponding binary linear code $C_{\D_{(a,b)}}$ is defined by
\begin{equation}\label{defcode}
C_{\D_{(a,b)}}=\Big\{\left( \Tr(\x\cdot \d_1), \Tr(\x\cdot \d_2),\ldots, \Tr(\x\cdot \d_{n})\right):\x\in \mathbb{F}_{q}^{2}\Big\}.
\end{equation}

Now, we first calculate the length of the linear codes $C_{\D_{(a,b)}}$ and the Hamming weight of non-zero codewords in $C_{\D_{(a,b)}}$.

\begin{lemma}\label{lem:n}
For $a\in\mathbb{F}_{q}^{\ast}, b\in\mathbb{F}_{q}$. Then,
$$
n=|\D_{(a,b)}|=\left\{\begin{array}{ll}
\frac{1}{2}q\big(q+1+S(a,0)\big)-1, & \textrm{if\ } \ b=0, \\
\frac{1}{2}q^{2}-1, & \textrm{if\ } \ b\neq0.
\end{array}
\right.
$$
\end{lemma}

\begin{proof}
By the orthogonal property of additive characters, we have
\begin{align*}
|\D_{(a,b)}|
&=\frac{1}{2}\sum_{x,y\in \mathbb{F}_{q}}\sum_{z\in \mathbb{F}_{2}}(-1)^{\mathrm{Tr}\big(z(ax^{\frac{q-1}{l^{m}}}+by)\big)}-1  \\
&=\frac{1}{2}\sum_{x,y\in \mathbb{F}_{q}}\big(1+(-1)^{\mathrm{Tr}(ax^{\frac{q-1}{l^{m}}}+by)}\big)-1  \\
&=\frac{1}{2}q^{2}+\frac{1}{2}\sum_{x,y\in \mathbb{F}_{q}}(-1)^{\mathrm{Tr}(ax^{\frac{q-1}{l^{m}}}+by)}-1  \\
&=\frac{1}{2}q^{2}+\frac{1}{2}\sum_{y\in \mathbb{F}_{q}}(-1)^{\mathrm{Tr}(by)}\sum_{x\in \mathbb{F}_{q}}(-1)^{\mathrm{Tr}(ax^{\frac{q-1}{l^{m}}})}-1.
\end{align*}

When $b=0$, we have
\begin{align*}
|\D_{(a,0)}|
&=\frac{1}{2}q^{2}+\frac{1}{2}q\sum_{x\in \mathbb{F}_{q}}(-1)^{\mathrm{Tr}(ax^{\frac{q-1}{l^{m}}})}-1  \\
&=\frac{1}{2}q^{2}+\frac{1}{2}q\big(1+S(a,0)\big)-1  \\
&=\frac{1}{2}q\big(q+1+S(a,0)\big)-1.
\end{align*}

When $b\neq0$, we have $\sum_{y\in \mathbb{F}_{q}}(-1)^{\mathrm{Tr}(by)}=0$, which follows that
$|\D_{(a,b)}|=\frac{1}{2}q^{2}-1$. We complete the proof.
\end{proof}

Let $\c_{(u,v)}$ be the corresponding codeword in $C_{\D_{(a,b)}}$ in \eqref{defcode} with $(u,v)\in \F_q^2$, that is,
$$
\c_{(u,v)}=\Big(\Tr(ux+vy)\Big)_{(x,y)\in D_{(a,b)}}.
$$
Obviously, $\c_{(0,0)} = 0$ and $\wt(\c_{(0,0)}) = 0$.
Now, we determine the Hamming weight $\wt(\c_{(u,v)})$ with $(u,v)\neq (0,0)$ in the following lemma.

\begin{lemma} \label{lem:wt}
Let $(u,v)(\neq (0,0))\in\F_{q}^{2}$. We have
\begin{enumerate}
\item[(1)] If $ b = 0 $, then
$$
\wt(\c_{(u,v)})=\left\{\begin{array}{ll}
\frac{1}{4}q\Big(q+1+S(a,0)\Big), \  & \textrm{if\ } \ v\neq0, \\
\frac{1}{4}q\Big(q+S(a,0)-S(a,u)\Big), \  & \textrm{if\ } \ \ v=0.
\end{array}
\right.
$$
\item[(2)] If $ b\neq 0$, then
$$
\wt(\c_{(u,v)})=\left\{\begin{array}{ll}
\frac{1}{4}q^{2}, \  & \textrm{if\ } \  v=0 \ \textrm{or\ } \ v\neq b,\ v\neq 0,\\
\frac{1}{4}q\Big(q-1-S(a,u)\Big), \  & \textrm{if\ } \ v=b.
\end{array}
\right.
$$
\end{enumerate}
\end{lemma}
\begin{proof}
Put $N(u,v)=\{(x,y)\in \mathbb{F}_{q}^{2}: \mathrm{Tr}(ax^{\frac{q-1}{l^{m}}}+by)=0, \mathrm{Tr}(ux+vy)=0\}$, then
\begin{align*}
|N(u,v)|
&=\frac{1}{4}\sum_{x,y\in \mathbb{F}_{q}}\Big(\sum_{z_{1}\in \mathbb{F}_{2}}(-1)^{\mathrm{Tr}(z_{1}(ax^{\frac{q-1}{l^{m}}}+by))}\sum_{z_{2}\in \mathbb{F}_{2}}(-1)^{\mathrm{Tr}(z_{2}(ux+vy))}\Big)  \\
&=\frac{1}{4}\sum_{x,y\in \mathbb{F}_{q}}\Big(\big(1+(-1)^{\mathrm{Tr}(ax^{\frac{q-1}{l^{m}}}+by)}\big)\big(1+(-1)^{\mathrm{Tr}(ux+vy)}\big)\Big)  \\
&=\frac{1}{4}q^{2}+\frac{1}{4}\sum_{x,y\in \mathbb{F}_{q}}(-1)^{\mathrm{Tr}(ax^{\frac{q-1}{l^{m}}}+by)}+\frac{1}{4}\sum_{x,y\in \mathbb{F}_{q}}(-1)^{\mathrm{Tr}(ax^{\frac{q-1}{l^{m}}}+ux+by+vy)},
\end{align*}
where we use the fact that
$$\sum_{x,y\in \F_{q}}(-1)^{\Tr(ux+vy)} =\sum_{x\in \F_{q}}(-1)^{\Tr(ux)}\sum_{y\in \F_{q}}(-1)^{\Tr(vy)} = 0.$$

Now we discuss case by case on the term of $b=0$ or $b\neq0$.

(1) If $b=0$, then
\begin{align*}
|N(u,v)|
&=\frac{1}{4}\Big(q^{2}+\sum_{x,y\in \F_{q}}(-1)^{\Tr(ax^{\frac{q-1}{l^{m}}})}+\sum_{x,y\in \F_{q}}(-1)^{\Tr(ax^{\frac{q-1}{l^{m}}}+ux+vy)}\Big)  \\
&=\frac{1}{4}\Big(q^{2}+q\sum_{x\in \F_{q}}(-1)^{\Tr(ax^{\frac{q-1}{l^{m}}})}+\sum_{y\in \F_{q}}(-1)^{\Tr(vy)}\sum_{x\in \F_{q}}(-1)^{\Tr(ax^{\frac{q-1}{l^{m}}}+ux)}\Big)  \\
&=\frac{1}{4}\Big(q^{2}+q(1+S(a,0))+(1+S(a,u))\sum_{y\in \F_{q}}(-1)^{\Tr(vy)}\Big).
\end{align*}
So, we have
$$
|N(u,v)|=\left\{\begin{array}{ll}
\frac{1}{4}q\Big(q+1+S(a,0)\Big), \  & \textrm{if\ } \ v\neq0, \\
\frac{1}{4}q\Big(q+2+S(a,0)+S(a,u)\Big), \  & \textrm{if\ } \  v=0.
\end{array}
\right.
$$

Noting that $\wt(\c_{(u,v)})=n-|N(u,v)|+1$.
By Lemma \ref{lem:n}, we have
$$
\wt(\c_{(u,v)})=\left\{\begin{array}{ll}
\frac{1}{4}q\Big(q+1+S(a,0)\Big), \  & \textrm{if\ } \ v\neq0, \\
\frac{1}{4}q\Big(q+S(a,0)-S(a,u)\Big), \  & \textrm{if\ } \ \ v=0.
\end{array}
\right.
$$

(2) If $b\neq0$, then $$\sum_{x,y\in \F_{q}}(-1)^{\Tr(ax^{\frac{q-1}{l^{m}}}+by)} =\sum_{x\in \F_{q}}(-1)^{\Tr(ax^{\frac{q-1}{l^{m}}})}\sum_{y\in \F_{q}}(-1)^{\Tr(by)} = 0 ,$$
which follows that
\begin{align*}
|N(u,v)|
&=\frac{1}{4}q^{2}+\frac{1}{4}\sum_{x,y\in \F_{q}}(-1)^{\Tr(ax^{\frac{q-1}{l^{m}}}+by+ux+vy)}  \\
&=\frac{1}{4}q^{2}+\frac{1}{4}\sum_{y\in \F_{q}}(-1)^{\Tr((b+v)y)}\sum_{x\in \F_{q}}(-1)^{\Tr(ax^{\frac{q-1}{l^{m}}}+ux)}  \\
&=\frac{1}{4}q^{2}+\frac{1}{4}(1+S(a,u))\sum_{y\in \F_{q}}(-1)^{\Tr((b+v)y)}.
\end{align*}

So, we have
$$
|N(u,v)|=\left\{\begin{array}{ll}
\frac{1}{4}q^{2}, \  & \textrm{if\ }  v=0, \\
\frac{1}{4}q\Big(q+1+S(a,u)\Big), \  & \textrm{if\ } \ v=b, \\
\frac{1}{4}q^{2}, \  & \textrm{if\ } \ v\neq b,\ v\neq0.
\end{array}
\right.
$$
By Lemma \ref{lem:n} again, we have
$$
\wt(\c_{(u,v)})=\left\{\begin{array}{ll}
\frac{1}{4}q^{2}, \  & \textrm{if\ } \ v=0, \\
\frac{1}{4}q\Big(q-1-S(a,u)\Big), \  & \textrm{if\ } \ v=b, \\
\frac{1}{4}q^{2}, \  & \textrm{if\ } \ v\neq b,\ v\neq0.
\end{array}
\right.
$$
The proof is finished.
\end{proof}

\begin{remark}\label{rem:dim}
By Lemma \ref{lem:wt}, we know that, for $(u,v)(\neq (0,0))\in\F_{q}^{2}$, we have $\wt(\c_{(u,v)})>0$. So, the map: $\F_q^2\rightarrow C_{\D_{(a,b)}}$ defined by $(u,v)\mapsto \c_{(u,v)} $
is an isomorphism as linear spaces over $\F_2$. Hence, the dimension of the codes $C_{\D_{(a,b)}}$ in \eqref{defcode} is equal to $2\phi(l^m)$.
\end{remark}

\begin{lemma}\label{lem:dis}
Let $C_{\D_{(a,b)}}$ be defined in \eqref{defcode}. Then, the minimal distance of the dual code $C_{{\D_{(a,b)}}}^{\perp}$ is at least $2$.
\end{lemma}
\begin{proof}
 We prove it by contradiction. If not, then there exists a coordinate $i$ such that the $i$-th entry of all of the codewords of $C_{\D_{(a,b)}}$ is $0$,
that is, $\Tr(\x\cdot \d_{i})=0$ for all $\x\in \F_{q}^{2}$, where $\d_{i}\in \D_{(a,b)}$.
Thus, by the properties of the trace function, we have $\d_{i}=0$. It contradicts with $\d_i \neq 0$.
\end{proof}

\subsection{Weight distribution of $C_{\D_{(a,b)}}$ in \eqref{defcode}}

In this sequel, we determine the weight distributions of linear codes $C_{\D_{(a,b)}}$ in \eqref{defcode}.

\begin{theorem}\label{thm:1}
Let $a\in\F_q^*$. The code $C_{\D_{(a,0)}}$ in \eqref{defcode} is an $[n,2\phi(l^{m}),d]$ binary linear code
with the weight distribution in Table 1, where $d=\frac{1}{4}q(q-\sqrt{q}+\frac{(q+\sqrt{q})S(a)}{l^{m}})$
and $n=\frac{1}{2}q\big(q+1+\frac{q-1}{l^m}S(a)\big)-1.$

\begin{table}[ht]
\centering
\caption{The weight distribution of the codes of Theorem 1.}
\begin{tabular}{|c|c|}
\hline
\textrm{Weight} $\omega$ \qquad& \textrm{Multiplicity} $A_\omega$   \\
\hline
0 \qquad&   1  \\
\hline
$\frac{1}{4}q(q+1+\frac{(q-1)S(a)}{l^{m}})$ \qquad&  $q(q-1)$  \\
\hline
$\frac{1}{4}q(q+\sqrt{q}+\frac{(q+\sqrt{q})S(a)}{l^{m}})$  \qquad& $\frac{1}{2}(q-1)(1+\frac{S(a)}{l^{m}})$  \\
\hline
$\frac{1}{4}q(q-\sqrt{q}+\frac{(q+\sqrt{q})S(a)}{l^{m}})$  \qquad& $\frac{1}{2}(q-1)(1-\frac{S(a)}{l^{m}})$  \\
\hline
\end{tabular}
\end{table}
\end{theorem}

\begin{proof}
Assume that $(u,v)\neq (0,0)$. By Lemma \ref{lem:S(a,b)} and Lemma \ref{lem:wt},
$\wt(\c_{(u,v)})$ has only three values, that is,
$$
\left\{\begin{array}{ll}
\omega_{1}=\frac{1}{4}q(q+1+\frac{(q-1)S(a)}{l^{m}}), \  \\
\omega_{2}=\frac{1}{4}q(q+\sqrt{q}+\frac{(q+\sqrt{q})S(a)}{l^{m}}), \  \\
\omega_{3}=\frac{1}{4}q(q-\sqrt{q}+\frac{(q+\sqrt{q})S(a)}{l^{m}}).
\end{array}
\right.
$$

Recall that $A_{\omega_{i}}$ is the multiplicity of $\omega_{i}$. By Lemma~\ref{lem:wt}, we have $A_{\omega_{1}}=q(q-1)$.
By Lemma~\ref{lem:dis} and the first two Pless Power Moment (\cite[P. 260]{12HP03} ),
we obtain the system of linear equations as follows:
$$
\left\{\begin{array}{ll}
A_{\omega_{1}}=q(q-1), \  \\
A_{\omega_{2}}+A_{\omega_{3}}=q-1, \  \\
\omega_{1}A_{\omega_{1}}+\omega_{2}A_{\omega_{2}}+\omega_{3}A_{\omega_{3}}=\frac{1}{2}q^{2}n.
\end{array}
\right.
$$
Solving the system , we get
$$
\left\{\begin{array}{ll}
A_{\omega_{1}}=q(q-1), \  \\
A_{\omega_{2}}=\frac{1}{2}(q-1)(1-\frac{S(a)}{l^{m}}), \  \\
A_{\omega_{3}}=\frac{1}{2}(q-1)(1+\frac{S(a)}{l^{m}}).
\end{array}
\right.
$$
Then we get the weight distribution of Table 1.
By Lemma~\ref{lem:3}, of all the non-zero weights, $\omega_{3}$ is the smallest. So we get the value of minimum Hamming weight $d=\omega_{3}$.
We complete the proof.
\end{proof}

\begin{corollary}\label{cor:1}
 For each $a\in \mathbb{F}_{q}^{\ast}, $ we have $|E_{a}|=\frac{1}{2}(q-1)(1+\frac{S(a)}{l^{m}})$ and $|O_{a}|=\frac{1}{2}(q-1)(1-\frac{S(a)}{l^{m}})$.
\end{corollary}
\begin{proof}  Recall the definitions of $E_{a}$ and $O_{a}$ in Subsection~\ref{subsec:3.1}. Obviously, $|E_{a}|+|O_{a}|=q-1$. By Lemma~\ref{lem:S(a,b)},
$$
S(a,u)=(-1)^{\wt((a u^{-\frac{q-1}{l^{m}}})^{(0)})}\sqrt{q}-\frac{\sqrt{q}+1}{l^{m}}S(a).
$$
For a code word $\c_{(u,v)} \in C_{\D_{(a,0)}}$, by the computing of $|N(u,v)|$, we know that $\wt(\c_{(u,v)})=\frac{1}{4}q(q-\sqrt{q}+\frac{(q+\sqrt{q})S(a)}{l^{m}})$
if and only if $2|\wt((a u^{-\frac{q-1}{l^{m}}})^{(0)})$. So the desired result follows from Theorem~\ref{thm:1}. The proof is completed.
\end{proof}

\begin{remark}
 By Lemma~\ref{lem:3}, we know that neither $E_{a}$ nor $O_{a}$ is empty for every $a\in \mathbb{F}_{q}^{\ast}$.
\end{remark}

\begin{example}
Let $(l,m,a,b)=(3,2,1,0).$ Then, the corresponding code $C_{\D_{(1,0)}}$ has parameters $[ 3199,12,1536]$ and weight enumerator
$1+49x^{1536}+4032x^{1600}+14x^{1792}$.
\end{example}

\begin{theorem}\label{thm:2}
Let $a, b\in\F_q^*$.
The code $C_{\D_{(a,b)}}$ in \eqref{defcode} is a $[\frac{1}{2}q^{2}-1,2\phi(l^{m})]$ binary linear code
with the weight distribution in Table 2.
\begin{table}[ht]
\centering
\caption{The weight distribution of the codes of Theorem 2.}
\begin{tabular}{|c|c|}
\hline
\textrm{Weight} $\omega$ \qquad& \textrm{Multiplicity} $A_\omega$   \\
\hline
0 \qquad&   1  \\
\hline
$\frac{1}{4}q(q-1-\frac{(q-1)S(a)}{l^{m}})$ \qquad&  $1$  \\
\hline
$\frac{1}{4}q(q-1-\sqrt{q}+\frac{(1+\sqrt{q})S(a)}{l^{m}})$  \qquad& $\frac{1}{2}(q-1)(1+\frac{S(a)}{l^{m}})$  \\
\hline
$\frac{1}{4}q(q-1+\sqrt{q}+\frac{(1+\sqrt{q})S(a)}{l^{m}})$  \qquad& $\frac{1}{2}(q-1)(1-\frac{S(a)}{l^{m}})$  \\
\hline
$\frac{1}{4}q^{2}$  \qquad& $q^{2}-q-1$  \\
\hline
\end{tabular}
\end{table}
\end{theorem}
\begin{proof}
Assume that $(u,v)\neq (0,0)$. By Lemma \ref{lem:S(a,b)} and Lemma \ref{lem:wt},
$\wt(\c_{(u,v)})$ has only four values, that is,
$$
\left\{\begin{array}{ll}
\omega_{1}=\frac{1}{4}q(q-1-\frac{(q-1)S(a)}{l^{m}}), \  \\
\omega_{2}=\frac{1}{4}q(q-1-\sqrt{q}+\frac{(1+\sqrt{q})S(a)}{l^{m}}), \  \\
\omega_{3}=\frac{1}{4}q(q-1+\sqrt{q}+\frac{(1+\sqrt{q})S(a)}{l^{m}}), \ \\
\omega_{4}=\frac{1}{4}q^{2}.
\end{array}
\right.
$$

By the computation of $|N(u,v)|$ in Lemma~\ref{lem:wt} and Corollary~\ref{cor:1}, it is easy to get the multiplicity $A_{\omega_{i}}$ of $\omega_{i}.$ They are listed as follows.
$$
\left\{\begin{array}{ll}
A_{\omega_{1}}=1, \  \\
A_{\omega_{2}}=\frac{1}{2}(q-1)(1+\frac{S(a)}{l^{m}}), \  \\
A_{\omega_{3}}=\frac{1}{2}(q-1)(1-\frac{S(a)}{l^{m}}), \\
A_{\omega_{4}}=q^{2}-q-1.
\end{array}
\right.
$$

Then we get the weight distribution of Table 2.
The proof is finished.
\end{proof}

\begin{example}
Let $(l,m,a,b)=(3,2,1,1)$. Then, the corresponding code $C_{\D_{(a,b)}}$ has parameters $[ 2047,12,448]$ and weight enumerator
$1+x^{448}+49x^{960}+14x^{1216}+4031x^{1024}$.
\end{example}

\begin{example}
Let $(l,m,a,b)=(5,1,1,1)$. Then, the corresponding code $C_{\D_{(a,b)}}$ has parameters $[ 127,8,32]$ and weight enumerator
$1+x^{96}+251x^{64}+3x^{32}$.
\end{example}

\begin{remark}  In Example~3, when $(l,m,a,b)=(5,1,1,1)$, the last two weights of $C_{\D_{(a,b)}}$ equal. So it is a $3$-weight code.
\end{remark}

\subsection{Weight hierarchy of $C_{\D_{(a,b)}}$ in \eqref{defcode}}

Let $H_{r}$ be an $r$-dimensional subspace of $\mathbb{F}_{q}^{2}$ and $\beta_{1},\beta_{2},\cdots,\beta_{r}$ be an $\mathbb{F}_{2}$-basis of $H_{r}$. We set
$$
N(H_{r})=\{\x=(x,y)\in \mathbb{F}_{q}^{2}: \mathrm{Tr}(ax^{\frac{q-1}{l^{m}}}+by)=0, \mathrm{Tr}(\x\cdot \beta_{j})=0, 1\leq j\leq r\}.
$$
Then, $N(H_{r}) = (\D_{(a,b)}\cap H_r^\perp)\cup\{(0,0)\}$, which concludes that $|N(H_{r})|=|\D_{(a,b)}\cap H_r^\perp|+1$.

So, by Remark~\ref{rem:dim} and the proof of proposition~\ref{pro:d_r}, we have
\begin{align}
    &d_{r}(C_{\D_{(a,b)}})\nonumber\\
    &= n-\max\Big\{|\D_{(a,b)}\cap H_{2\phi(l^m)-r}|: H_{2\phi(l^m)-r} \in [\mathbb{F}_{q}^{2},2\phi(l^m)-r]_{p}\Big\}\label{eq:4}\\
    &=n-\max\Big\{|N(H_{r})|: H_{r} \in [\mathbb{F}_{q}^{2},r]_{p}\Big\}+1.
\end{align}
\begin{proposition} \label{pro:d_r:2}
Define $B_{H_{r}}=\sum_{(x,y)\in \F_{q}^{2}}\sum_{\beta\in H_{r}}(-1)^{\Tr(\beta\cdot (x,y)+ax^{\frac{q-1}{l^{m}}}+by)}$. Then,
\begin{align*}
d_{r}(C_{\D_{(a,b)}})= n-\frac{1}{2^{r+1}}q^2-\frac{1}{2^{r+1}}\max\Big\{B_{H_{r}}: H_{r} \in [\mathbb{F}_{q}^{2},r]_{p}\Big\}+1.
\end{align*}
\end{proposition}
\begin{proof}
By the orthogonal property of additive characters, we have
\begin{align*}
&2^{r+1}|N(H_{r})|\\
&=\sum_{\x=(x,y)\in \F_{q}^{2}}\Big(\sum_{z\in \F_{2}}(-1)^{\Tr(z(ax^{\frac{q-1}{l^{m}}}+by))}\prod_{i=1}^{r}\sum_{x_{i}\in \F_{2}}(-1)^{\Tr(x_{i}(\x\cdot \beta_{i}))}\Big)  \\
&=\sum_{\x=(x,y)\in \F_{q}^{2}}\Big(\sum_{z\in \F_{2}}(-1)^{\Tr(z(ax^{\frac{q-1}{l^{m}}}+by))}\sum_{\beta\in H_{r}}(-1)^{\Tr(\beta\cdot \x)}\Big)  \\
&=\sum_{\x=(x,y)\in \F_{q}^{2}}\Big(\big(1+(-1)^{\Tr(ax^{\frac{q-1}{l^{m}}}+by)}\big)\sum_{\beta\in H_{r}}(-1)^{\Tr(\beta\cdot \x)}\Big) \\
&=\sum_{\x=(x,y)\in \F_{q}^{2}}\sum_{\beta\in H_{r}}(-1)^{\Tr(\beta\cdot \x)}+\sum_{\x=(x,y)\in \F_{q}^{2}}\sum_{\beta\in H_{r}}(-1)^{\Tr(\beta\cdot \x+ax^{\frac{q-1}{l^{m}}}+by)} \\
&=q^2+\sum_{\x=(x,y)\in \F_{q}^{2}}\sum_{\beta\in H_{r}}(-1)^{\Tr(\beta\cdot \x+ax^{\frac{q-1}{l^{m}}}+by)},
\end{align*}
where the last equation comes from
\begin{align*}
\sum_{\x=(x,y)\in \F_{q}^{2}}\sum_{\beta\in H_{r}}(-1)^{\Tr(\beta\cdot \x)}
&=\sum_{\x=(x,y)\in \F_{q}^{2}} \ 1+\sum_{\x=(x,y)\in \F_{q}^{2}}\sum_{(0,0)\neq\beta\in H_{r}}(-1)^{\Tr(\beta\cdot \x)}  \\
&=q^{2}+\sum_{(0,0)\neq\beta\in H_{r}}\sum_{\x=(x,y)\in \F_{q}^{2}}(-1)^{\Tr(\beta\cdot \x)}=q^{2}.
\end{align*}
So, the desired result is obtained. Thus, we complete the proof.
\end{proof}

In the following sequel, we shall determine the weight hierarchy of $C_{\D_{(a,b)}}$ by calculating $B_{H_{r}}$ in Proposition~\ref{pro:d_r:2} and $|\D_{(a,b)}\cap H_{2\phi(l^m)-r}|$ in \eqref{eq:4}.

\begin{theorem}\label{thm:3}
Let $a\in\F_q^*$ and $C_{\D_{(a,0)}}$ defined in \eqref{defcode}. Then
$$
d_{r}(C_{\D{(a,0)}})=\left\{\begin{array}{ll}
\frac{1}{2}q\Big(1-\frac{1}{2^{r}}\Big)\Big(q-\sqrt{q}+\frac{(q+\sqrt{q})S(a)}{l^{m}}\Big), \textrm{if\ } \ 1\leq r \leq \frac{1}{2}\phi(l^{m}), \\
\frac{1}{2}q\Big(q+1+\frac{(q-1)S(a)}{l^{m}}\Big)-\frac{1}{2^r}q^{2},  \textrm{if\ } \ \frac{1}{2}\phi(l^{m}) < r \leq 2\phi(l^{m}).
\end{array}
\right.
$$
\end{theorem}
\begin{proof}
When $1 \leq r \leq \frac{1}{2}\phi(l^{m})$, by the definition of $B_{H_{r}}$, we have
\begin{align*}
&B_{H_{r}}
=\sum_{\x=(x,y)\in \F_{q}^{2}}\sum_{\beta\in H_{r}}(-1)^{\Tr(\beta\cdot \x+ax^{\frac{q-1}{l^{m}}})}  \\
&=\sum_{(x,y)\in \F_{q}^{2}}\Big(\sum_{(\beta_{1},0)\in H_{r}}(-1)^{\Tr(\beta_{1}x+ax^{\frac{q-1}{l^{m}}})}+\sum_{\substack{(\beta_{1},\beta_{2})\in H_{r}\\\beta_{2}\neq 0}}(-1)^{\Tr(\beta_{1}x+\beta_{2}y+ax^{\frac{q-1}{l^{m}}})}\Big)  \\
&=q\sum_{(\beta_{1},0)\in H_{r}}\sum_{x\in \F_{q}}(-1)^{\Tr(\beta_{1}x+ax^{\frac{q-1}{l^{m}}})}\\
&=q\sum_{(\beta_{1},0)\in H_{r}}\Big(1+S(a,\beta_{1})\Big).
\end{align*}
So, by lemma~\ref{lem:S(a,b)}, we obtain
\begin{align*}
\frac{1}{q}B_{H_{r}}
&=\sum_{(\beta_{1},0)\in H_{r}}\Big(1-\frac{(\sqrt{q}+1)S(a)}{l^{m}}\Big)+\sqrt{q}\sum_{\substack{(\beta_{1},0)\in H_{r}\\ \beta_{1}\neq 0}}(-1)^{\wt((a\beta_{1}^{-\frac{q-1}{l^{m}}})^{(0)})}   \\
&+\frac{(\sqrt{q}+q)S(a)}{l^{m}}.
\end{align*}

By Corollary~\ref{cor:1}, there is an element $\beta \in \F_{q}^{\ast}$ such that $\wt((a\beta^{-\frac{q-1}{l^{m}}})^{(0)})$ is even.
Recall that $2$ is a primitive root modulo $l^{m}$ and $q=2^{\phi(l^m)}$, we have $q \equiv 1 \mod l^{m}$, which concludes that $\sqrt{q} \equiv -1 \mod l^m$, i.e. $l^m|\sqrt{q}+1$.
Since $\sqrt{q}-1$ and $ \sqrt{q}+1$ are coprime, we have $(\beta u)^{-\frac{q-1}{l^m}}=\beta^{-\frac{q-1}{l^m}}(u^{\sqrt{q}-1})^{-\frac{\sqrt{q}+1}{l^m}} =\beta^{-\frac{q-1}{l^m}}$ for any $u\in\F_{\sqrt{q}}^*$.
Take an $r$-dimensional subspace $L_{r}$ contained in $\beta\F_{\sqrt{q}}$ and
Put $H_{r}=L_{r}\times O$, then, for any $(\beta_1,0)\in H_r$, we have $\wt((a\beta_{1}^{-\frac{q-1}{l^{m}}})^{(0)})$ is even.
Furthermore, by Lemma~\ref{lem:3}, we have $-\sqrt{q}<1-\frac{(\sqrt{q}+1)S(a)}{l^{m}}<\sqrt{q}$, which follows that $1-\frac{(\sqrt{q}+1)S(a)}{l^{m}}+\sqrt{q}>0$. Hence,
$\frac{1}{q}B_{H_{r}}$ reaches its maximum
$$
2^{r}\Big(1-\frac{(\sqrt{q}+1)S(a)}{l^{m}}+\sqrt{q}\Big)+\frac{(\sqrt{q}+q)S(a)}{l^{m}}-\sqrt{q}.
$$

So, by Proposition~\ref{pro:d_r:2}, we obtain the generalized Hamming weights $d_{r}(C_{\D_{(a,0)}})$ for $1 \leq r \leq \frac{1}{2}\phi(l^{m})$.

When $ \frac{1}{2}\phi(l^{m})< r<2 \phi(l^{m})$, we have $ 1\leq 2 \phi(l^{m})-r<\frac{3}{2}\phi(l^{m})$.
By Lemma~\ref{lem:n} and Lemma~\ref{lem:3}, there exists an element $(x,y)\in \D_{(a,0)}$ such that $x\in\F_{q}^{\ast}$,
that is, $\Tr(ax^{\frac{q-1}{l^{m}}})=0$ and $x \neq 0$. Then, for any $u\in\F_{\sqrt{q}}^*$, we have $\Tr(a(xu)^{\frac{q-1}{l^{m}}})=\Tr(ax^{\frac{q-1}{l^{m}}})=0$.
So, $x\F_{\sqrt{q}}\times\F_{q}\subset \D_{(a,0)}$.
Note that the dimension of the subspace $x\F_{\sqrt{q}}\times\F_{q}\subset\F_{q}^{2}$ is $\frac{3}{2}\phi(l^{m})$.
Let $H_{2 \phi(l^{m})-r}$ be a $(2 \phi(l^{m})-r)$-dimensional subspace of $x\F_{\sqrt{q}}\times\F_{q}$.
So,
$$
|H_{2 \phi(l^{m})-r}\cap \D_{(a,0)}|=2^{2 \phi(l^{m})-r}-1,
$$
Then,
$$
\max\{|\D_{(a,0)} \cap H|: H \in [\F_{p^{m}},2\phi(l^{m})-r]_{p}\}=2^{2 \phi(l^{m})-r}-1.
$$
By the equation \eqref{eq:d_r}, for $ \frac{1}{2}\phi(l^{m})< r<2 \phi(l^{m})$, we have
\begin{align*}
   d_{r}(C_{\D_{(a,0)}})&=n-2^{2 \phi(l^{m})-r}+1\\
&=\frac{1}{2}q\Big(q+1+\frac{(q-1)S(a)}{l^{m}}\Big)-\frac{1}{2^r}q^{2}.
\end{align*}
Thus, we complete the proof.
\end{proof}

\begin{theorem}\label{thm:4}
Let $a, b\in\F_q^*$ and $C_{\D_{(a,b)}}$ defined in \eqref{defcode}. Then,
\begin{itemize}
\item[(1)] If $ 1\leq r \leq \frac{1}{2}\phi(l^{m})$, we have
$$
d_{r}(C_{\D_{(a,b)}})=\left\{\begin{array}{ll}
\frac{1}{2}q^{2}\Big(1-\frac{1}{2^r}\Big)-\frac{q(\sqrt{q}+1)}{4}\Big(1-\frac{S(a)}{l^{m}}\Big), & \textrm{if\ } \ S(a)<0, \\
\frac{1}{2}q^{2}\Big(1-\frac{1}{2^r}\Big)-\frac{q(\sqrt{q}+1)}{4}\Big(1-\frac{S(a)}{l^{m}}\Big)&+\frac{q\sqrt{q}}{2^{r+1}}\Big(1-\frac{\sqrt{q}+1}{l^m}S(a)\Big),\\ &\textrm{if\ } \ S(a)>0.
\end{array}
\right.
$$
\item[(2)] If $\frac{1}{2}\phi(l^{m}) < r \leq 2\phi(l^{m})$, we have
$$d_{r}(C_{\D_{(a,b)}})=\frac{1}{2}q^2(1-\frac{1}{2^{r-1}}).$$
\end{itemize}
\end{theorem}

\begin{proof}
(1) When $1 \leq r \leq \frac{1}{2}\phi(l^{m})$, by the definition of $B_{H_{r}}$, we have
\begin{align*}
B_{H_{r}}
&=\sum_{(x,y)\in \F_{q}^{2}}\sum_{(\beta_{1},\beta_{2})\in H_{r}}(-1)^{\Tr(ax^{\frac{q-1}{l^{m}}}+\beta_{1}x+\beta_{2}y+by)}   \\
&=\sum_{(\beta_{1},\beta_{2})\in H_{r}}\sum_{x\in \F_{q}}(-1)^{\Tr(ax^{\frac{q-1}{l^{m}}}+\beta_{1}x)}\sum_{y\in \F_{q}}(-1)^{\Tr(\beta_{2}y+by)}.
\end{align*}

Let $ \Prj_{2}$ be the second projection from $\F_{q}^{2}$ to $\F_{q}$ defined by $(x,y)\mapsto y$.

If $b \notin \Prj_{2}(H_{r})$, then $\sum_{y\in \F_{q}}(-1)^{\Tr(\beta_{2}y+by)} =0$ for any $(\beta_{1},\beta_{2})\in H_{r}$, which follows that $B_{H_{r}}=0$.

If $b \in \Prj_{2}(H_{r})$, by Lemma~\ref{lem:S(a,b)} again, we have
\begin{align*}
&\frac{1}{q}B_{H_{r}}
=\sum_{(\beta_{1},b)\in H_{r}}\sum_{x\in \F_{q}}(-1)^{\Tr(ax^{\frac{q-1}{l^{m}}}+\beta_{1}x)}
=\sum_{(\beta_{1},b)\in H_{r}}1+\sum_{(\beta_{1},b)\in H_{r}}S(a,\beta_{1})\\
\end{align*}
\begin{equation*}
\begin{split}
&=\left\{\begin{array}{ll}
\sum\limits_{\substack{(\beta_{1},b)\in H_{r}\\\beta_1\neq 0}}\Big(1-\frac{\sqrt{q}+1}{l^m}S(a)+(-1)^{\wt((a\beta_{1}^{-\frac{{q-1}}{l^m}})^{(0)})}\sqrt{q}\Big),
 &\textit{if $(0,b)\notin H_{r}$},\\
\sum\limits_{\substack{(\beta_{1},b)\in H_{r}\\\beta_1\neq 0}}\Big(1-\frac{\sqrt{q}+1}{l^m}S(a)+(-1)^{\wt((a\beta_{1}^{-\frac{{q-1}}{l^m}})^{(0)})}\sqrt{q}\Big)&+1+\frac{q-1}{l^m}S(a),\\
&\textit{if $(0,b)\in H_{r}$}.
\end{array}
\right.
\end{split}
\end{equation*}

By Lemma~\ref{lem:3}, we have $1-\frac{(\sqrt{q}+1)S(a)}{l^{m}}+\sqrt{q}>0$. Take an element $\beta \in \F_{q}^{\ast}$
so that $\wt((a\beta^{-\frac{q-1}{l^{m}}})^{(0)})$ is even and an $(r-1)$-dimensional space $L_{r-1}$ of $\beta\F_{\sqrt{q}}$.
Then, we have $(\beta u)^{-\frac{q-1}{l^m}}=\beta^{-\frac{q-1}{l^m}}(u^{\sqrt{q}-1})^{-\frac{\sqrt{q}+1}{l^m}} =\beta^{-\frac{q-1}{l^m}}$ for any $u\in\F_{\sqrt{q}}^*$.
So, for any non-zero element $\beta_1\in L_{r-1}$, $\wt((a\beta_1 ^{-\frac{q-1}{l^{m}}})^{(0)})$ is even.

If $S(a)<0 $, set $H_{r}= \{(u,0):u\in L_{r-1}\}\cup \{(\xi+u,b):u\in L_{r-1}\}$, where $\xi \in \beta\F_{\sqrt{q}}\backslash L_{r-1}$. Then, $(0,b)\notin H_{r}$.
In this case, $\frac{1}{q}B_{H_{r}}$ reaches its maximum
$$2^{r-1}\Big(1-\frac{(\sqrt{q}+1)S(a)}{l^{m}}+\sqrt{q}\Big).$$

If $S(a)>0 $, set $H_{r}= L_{r-1}\times b\F_{2}$. Then, $(0,b)\in H_{r}$.
In this case, $\frac{1}{q}B_{H_{r}}$ reaches its maximum
$$2^{r-1}(1-\frac{(\sqrt{q}+1)S(a)}{l^{m}}+\sqrt{q})+\frac{(\sqrt{q}+q)S(a)}{l^{m}}-\sqrt{q}.$$

So, by Proposition~\ref{pro:d_r:2}, we obtain the generalized Hamming weights $d_{r}(C_{\D_{(a,b)}})$ for $1 \leq r \leq \frac{1}{2}\phi(l^{m})$.

(2) For  $\frac{1}{2}\phi(l^{m})< r\leq 2\phi(l^{m})$, we have $1\leq2\phi(l^{m})-r\leq \frac{3}{2}\phi(l^{m})-1$.
Take an element $x\in\F_{q}^{\ast}$ such that $\Tr(ax^{\frac{q-1}{l^{m}}})=0$. Then, for any $u\in\F_{\sqrt{q}}^*$, we have $\Tr(a(xu)^{\frac{q-1}{l^{m}}})=\Tr(ax^{\frac{q-1}{l^{m}}})=0$.
Let $T_{b}=\{y\in \F_{q}: \mathrm{Tr}(by)=0\}$, which is a $(\phi(l^{m})-1)$-dimensional subspace.
So, $x\F_{\sqrt{q}}\times T_{b}\subset \D_{(a,b)}$.
Note that the dimension of the subspace $x\F_{\sqrt{q}}\times T_{b}\subset\F_{q}^{2}$ is $\frac{3}{2}\phi(l^{m}) - 1$.
Let $H_{2 \phi(l^{m})-r}$ be a $(2 \phi(l^{m})-r)$-dimensional subspace of $x\F_{\sqrt{q}}\times T_{b}$.
So,
$$
|H_{2 \phi(l^{m})-r}\cap \D_{(a,b)}|=2^{2 \phi(l^{m})-r}-1,
$$
Then,
$$
\max\{|\D_{(a,b)} \cap H|: H \in [\F_{p^{m}},2\phi(l^{m})-r]_{p}\}=2^{2 \phi(l^{m})-r}-1.
$$
By the equation \eqref{eq:d_r}, for $ \frac{1}{2}\phi(l^{m})< r<2 \phi(l^{m})$, we have
$$
d_{r}(C_{\D_{(a,b)}})=n-\frac{1}{2^r}q^2=\frac{1}{2}q^2(1-\frac{1}{2^{r-1}}).
$$
Thus, we complete the proof.
\end{proof}

\section{Concluding Remarks}
The method of constructing linear codes by defining sets can be generalized. In this paper,
we give a formula for calculating the weight hierarchies of linear codes constructed by the generalizes method.
Then we construct two classes of 3-weight or 4-weight binary linear codes based on the generalized method.
By the exponential sum theory, we give the weight distributions of these codes. Using our formula, we determine the
weight hierarchies of these codes completely.

Let $\omega_{\min}$ and $\omega_{\max}$ be the minimum and maximum nonzero weight of the linear code $C_{\D_{(a,0)}}$, respectively.
It is easy to check that
$$
 \frac{\omega_{\min}}{\omega_{\max}}> \frac{1}{2}.
$$
By the results in \cite{21YD06}, the binary linear codes $C_{\D_{(a,0)}}$ in Theorem 1
are suitable for constructing secret sharing schemes with interesting properties.

\section*{Acknowledgement}

For the research, the first author was supported by the National Science Foundation of China Grant No. 11701001
and  Anhui Provincial Natural Science Foundation No. 1908085MA02,
and the second author was supported by the National Science Foundation of China Grant No. 11701317.

\end{document}